\RequirePackage{amsmath}
\documentclass[runningheads]{llncs}

\usepackage{graphicx}
\usepackage{url}
\usepackage{cite}
\usepackage{todonotes}
\usepackage{algorithm}
\usepackage{algpseudocode}
\usepackage{mathtools}
\usepackage{hyperref}
\usepackage[capitalise]{cleveref}
\usepackage{multirow}
\usepackage{verbatim}
\usepackage{amsfonts}
\usepackage{placeins}
\usepackage{tabularx}

\definecolor{col2}{HTML}{ffce73}
\definecolor{col3}{HTML}{779e11}
\definecolor{col1}{HTML}{0d5273}

\usetikzlibrary{positioning, calc, arrows, arrows.meta}

\newcommand{\OhOfB}[1]{\ensuremath{\mathcal{O}(#1)}}
\newcommand{\OhOf}[1]{\ensuremath{\OhOfB{#1}}}
\newcommand{\OhOfExp}[1]{\ensuremath{\smash{\hat{\mathcal{O}}}(#1)}}
\newcommand{\OhOfLarge}[1]{\ensuremath{{\mathcal{O}\left(#1\right)}}}
\newcommand{\ellmax}{\ensuremath{d}}

\newcommand{\Lrelax}{\ensuremath{{L'}}}

\newcommand{\Zq}{\mathbb{Z}_q}
\newcommand{\ceil}[1]{\left\lceil#1\right\rceil}
\newcommand{\absolute}[1]{\left\lvert#1\right\rvert}
\newcommand{\llex}{\prec}
\newcommand{\leqlex}{\preceq}
\newcommand{\tuple}[1]{\left\langle#1\right\rangle}

\usepackage{microtype}

\begin{document}
\title{LCP-Aware Parallel String Sorting}
%
%
%
\author{Jonas Ellert\inst{1} \and
Johannes Fischer\inst{1} \and
Nodari Sitchinava\thanks{Supported by the National Science Foundation under Grants CCF-1533823 and CCF-1911245}
\inst{2}}
\authorrunning{Ellert et al.}
%
\institute{%
Technical University of Dortmund,
{Dortmund, Germany}\\
\email{jonas.ellert@tu-dortmund.de}, \email{johannes.fischer@cs.tu-dortmund.de} 
\and
University of Hawaii at Manoa,
{Honolulu, Hawaii, USA}\\
\email{nodari@hawaii.edu}}
\maketitle              
\begin{abstract}
When lexicographically sorting strings, it is not always necessary to inspect all symbols. 
For example, the lexicographical rank of \texttt{europar} amongst the strings \texttt{eureka}, \texttt{eurasia}, and \texttt{excells} only depends on its so called \emph{relevant prefix} \texttt{euro}.
The \emph{distinguishing prefix size} $D$ of a set of strings is the number of symbols that actually need to be inspected to establish the lexicographical ordering of all strings.
Efficient string sorters should be \emph{D}-aware, i.e.\ their complexity should depend on $D$ rather than on the total number $N$ of all symbols in all strings.
While there are many $D$-aware sorters in the sequential setting, there appear to be no such results in the PRAM model.
We propose a framework yielding a $D$-aware modification of any existing PRAM string sorter. 
The derived algorithms are work-optimal with respect to their original counterpart: 
If the original algorithm requires $\OhOf{w(N)}$ work, the derived one requires $\OhOf{w(D)}$ work. 
The execution time increases only by a small factor that is logarithmic in the length of the longest relevant prefix.
Our framework universally works for deterministic and randomized algorithms in all variations of the PRAM model, such that future improvements in ($D$-unaware) parallel string sorting will directly result in improvements in $D$-aware parallel string sorting.

\end{abstract}

\section{Introduction}

The problem of string sorting is defined as follows: Given $k$ strings $s_1, \dots, s_k$ of total length $N=\sum |s_i|$ stored in RAM, and an array $S$ of $k$ pointers to the strings ($S[i]$ points to the memory location of $s_i$), compute a permutation $S'$ of $S$ such that $S'$ lists the strings in lexicographical order ($S'[i]$ points to the lexicographically $i$-th smallest string).
It is commonly known that establishing the lexicographical order on the strings does not necessarily require inspecting all $N$ symbols.
In fact, the rank of a string $s_i$ only depends on its shortest prefix $s_i[1..\ell_i]$ that is not a prefix of any another string.
The \emph{distinguishing prefix size} of the $k$ strings is defined as $D = \sum_{i = 1}^k \ell_i$. In simple words, an algorithm that sorts the strings only needs to \emph{inspect} the $D$ symbols that are part of the distinguishing prefix, while all other symbols are irrelevant for the lexicographical ordering.
In this paper, we present parallel $D$-aware string sorting solutions. 
That is, the time and work complexity of the algorithms depends only on $k$, $D$, and possibly $\sigma$, but not on $N$.
We present algorithms in the PRAM model and consider the following variations of the model (ordered from the weakest to the strongest): EREW, CREW, Common-CRCW, and Arbitrary-CRCW. Observe that algorithms designed for the weaker models can run on the stronger models within the same complexity measures.

\subsection{Related Work}
\label{sec:relwork}

There is a variety of algorithms that aim to efficiently solve the problem of string sorting, most of which belong to one of two classes:
The ones that are based on comparison sorting and generally allow arbitrary alphabets, and the ones that use (ideas from) integer sorting and are usually limited to alphabets of polynomial size $\sigma = N^{\OhOf{1}}$.

If comparison sorting is the underlying technique, the well-known information-theoretical lower bound of $\Omega(k \lg k)$ comparisons applies, such that the fastest possible sequential algorithm cannot take fewer than $\Omega(k \lg k + D)$ operations. 
Ternary quicksort \cite{DBLP:conf/soda/BentleyS97} runs in $\OhOf{k\lg k + D}$ time, and thus matches this lower bound. 
In the Common-CRCW model, J\'aJ\'a et al.~\cite{DBLP:journals/tcs/JaJaRV96} achieve $\OhOf{k \lg k + N}$ work and $\OhOf{\lg^2 k / \lg\lg k}$ time, and also provide a randomized algorithm that requires the same amount of work and $\OhOf{\lg k}$ time with high probability. 
However, a $D$-aware modification of the algorithm cannot easily be derived.

In terms of alphabet-dependent sequential algorithms, we can use radix-sort-like approaches to achieve either $\OhOf{N+\sigma}$ time \cite[Alg.\ 3.2]{DBLP:books/aw/AhoHU74}, or even $\OhOf{D+\sigma}$ time \cite{DBLP:journals/siamcomp/PaigeT87}, where $\sigma$ is the number of different characters.
Hagerup~\cite{DBLP:conf/stoc/Hagerup94} presents an Arbitrary-CRCW algorithm that achieves $\OhOf{N \lg \lg N}$ work and $\OhOf{\lg N / \lg \lg N}$ time, assuming that the alphabet is polynomial in $N$. 
Alternatively, it can be implemented to run in
$\OhOf{N \sqrt{\lg N}}$ work and $\OhOf{\lg^{3/2} N\sqrt{\lg\lg N}}$ time in the CREW model, or $\OhOf{N \sqrt{\lg N\lg\lg N}}$ work and the same time in the EREW model. 
Note that Hagerup's algorithm is based on an algorithm by Vaidyanathan et al. \cite{Vaidyanathan1991} that reduces each string to a single integer by repeatedly merging adjacent symbols. 
Due to the nature of the reduction technique, it always inspects all $N$ symbols, and a $D$-aware modification cannot easily be derived. 

There are practical parallel algorithms that exploit the distinguishing prefix and are fast in practice \cite{Bingmann2013,Bingmann2017,Bingmann2020}; however, we are not aware of any algorithms with $D$-aware complexity bounds in the PRAM model.

\subsection{Our Contributions}

We present a theoretical framework that yields a $D$-aware version of any existing string sorting algorithm.
Particularly, we derive $D$-aware versions of the algorithms by J\'aJ\'a et al. and Hagerup that are work optimal with respect to their original counterparts: 
If the original algorithm requires $w(k, N, \sigma)$ work, then our modification requires $\OhOf{w(k, D, \sigma)}$ work.
Additionally, in case of Hagerup's algorithm, we are no longer limited to polynomial alphabets. 
Generally, the new algorithms are only by a $(\lg d)$-factor slower than the original ones, where $d = \max\{\ell_i \mid 1 \leq i \leq k\}$ denotes the length of the longest relevant prefix.

Our framework is based on the idea of approximating the distinguishing prefix. 
It yields a $2$-approximation of the relevant prefix lengths: For each string $s_i$, we determine a value $L[i] \in [\ell_i, 2\ell_i)$.
In the Arbitrary-CRCW model, this takes expected optimal $\OhOf{D}$ work and $\OhOf{\lg d \cdot (\lg d + \lg k)}$ time with high probability%
. 
In the weaker EREW model, we achieve ${\OhOf{k\sqrt{\lg k}\lg\lg k + D}}$ work and $\OhOf{\lg d \cdot (\lg d + \lg k) + \lg^{3 / 2} k \cdot \lg\lg k}$ time with high probability%
.
An overview of our results is provided in \cref{tbl:results}.

\begin{table}[t]
	\newcolumntype{Y}{>{\centering\arraybackslash}X}
	\setlength\tabcolsep{3.5pt}
	\caption{New results on $D$-aware parallel string sorting.
	The original ($D$-unaware) results are written in gray.
	Whenever the model is annotated with \emph{w.h.p.}, the respective algorithms are successful with high probability $1 - \OhOf{k^{-c}}$ for an arbitrarily large constant $c$.
	We write $\OhOfExp{x}$ to denote \emph{expected} complexity bounds.}
	\newcommand{\upperspace}{\vphantom{$\overset{T}{T}$}}
	\newcommand{\lowerspace}{\vphantom{$\underset{T}{T}$}}
	\newcommand{\tablehead}{\lowerspace}
	\newcommand{\factorshow}{{\lg \ellmax\ \cdot\hspace{-.5em}}}
	\newcommand{\factorhide}{\phantom{\factorshow}}
	\begin{tabularx}{\textwidth}{|Y|l|rl|c|}
		\multicolumn{5}{c}{\tablehead\textbf{a.) Results based on the sorter by Hagerup \cite{DBLP:conf/stoc/Hagerup94}:}}\\
		\hline
		\textbf{Model} & \multicolumn{1}{c|}{\textbf{Work}} & \multicolumn{2}{c|}{\textbf{Time}}& \textbf{Theorem} \\\hline
		\multirow{2}{*}{$\overset{\text{Arbitrary}}{\text{CRCW}}$} & 
		$\OhOf{D \lg \lg \max(D, \sigma))}$ & 
		$\factorshow$ & $\OhOf{\lg D / \lg \lg D + \lg \lg \sigma}$ & 
		\cref{lemma:hagerupfinal}\upperspace\\
		& \color{gray}$\OhOf{N \lg \lg N}$ &
		& \color{gray}$\OhOf{\lg N / \lg \lg N}$ & 
		\color{gray}\cite{DBLP:conf/stoc/Hagerup94} Theorem 4.4\lowerspace\\
		\multirow{2}{*}{CREW} &
		$\OhOf{D \sqrt{\lg D}}$ & 
		$\factorshow$ & $\OhOf{\lg^{3/2} D \sqrt{\lg \lg D}}$ & 
		\cref{lemma:hagerupfinal}\upperspace\\
		& \color{gray}$\OhOf{N \sqrt{\lg N}}$ &
		& \color{gray}$\OhOf{\lg^{3/2} N \sqrt{\lg \lg N}}$ & 
		\color{gray}\cite{DBLP:conf/stoc/Hagerup94} Theorem 4.5\lowerspace\\
		\multirow{2}{*}{EREW} & 
		$\OhOf{D \sqrt{\lg D \lg \lg D}}$ & 
		$\factorshow$ & $\OhOf{\lg^{3/2} D \sqrt{\lg \lg D}}$ & 
		\cref{lemma:hagerupfinal}\upperspace\\
		& \color{gray}$\OhOf{N \sqrt{\lg N \lg \lg N}}$ &
		& \color{gray}$\OhOf{\lg^{3/2} N \sqrt{\lg \lg N}}$ & 
		\color{gray}\cite{DBLP:conf/stoc/Hagerup94} Theorem 4.5\lowerspace\\\hline
		\multicolumn{5}{c}{}\\
		\multicolumn{5}{c}{\tablehead\textbf{b.) Results based on the sorter by J\'aJ\'a et al. \cite{DBLP:journals/tcs/JaJaRV96}:}}\\
		\hline
		\textbf{Model} & \multicolumn{1}{c|}{\textbf{Work}} & \multicolumn{2}{c|}{\textbf{Time}} & \textbf{Theorem} \\\hline
		\multirow{2}{*}{$\overset{\text{Common}}{\text{CRCW}}$} & 
		$\OhOf{k \lg k + D}$ & 
		$\factorshow$ & $\OhOf{\lg^2 k/\lg\lg k}$ &
		\cref{lemma:jajaplain}\upperspace\\
		& \color{gray}$\OhOf{k \lg k + N}$ &
		& \color{gray}$\OhOf{\lg^2 k/\lg\lg k}$ &
		\textcolor{gray}{\cite{DBLP:journals/tcs/JaJaRV96} Theorem 3.1}\lowerspace\\
		\multirow{2}{*}{$\underset{\text{w.h.p.}}{\overset{\text{Common}}{\text{CRCW}}}$} & 
		$\OhOf{k \lg k + D}$ & $\factorshow$ & 
		$\OhOf{\lg k + \lg \ellmax}$ & 
		\cref{lemma:jajarand}\upperspace\\
		& \color{gray}$\OhOf{k \lg k + N}$ &
		& \color{gray}$\OhOf{\lg k}$ &
		\textcolor{gray}{\cite{DBLP:journals/tcs/JaJaRV96} Theorem 5.1}\lowerspace\\\hline
		\multicolumn{5}{c}{}\\
	\end{tabularx}
	
	\begin{tabularx}{\textwidth}{|Y|l|rl|Y|}
		\multicolumn{5}{c}{\tablehead\textbf{c.) General results that hold for any parallel string sorter:}}\\
		\hline
		\textbf{Model} & \multicolumn{1}{c|}{\textbf{Work}} & \multicolumn{2}{c|}{\textbf{Time}}&\textbf{Lemma} \\\hline
		\multirow{2}{*}{$\underset{\text{w.h.p.}}{\overset{\text{Arbitrary}}{\text{CRCW}}}$} & 
		$\OhOfExp{D} + w(k, 2D, \sigma)$ & 
		$\factorshow$ & $\OhOf{\lg k + \lg \ellmax} + t(k, 2D, \sigma)$ & 
		\cref{lemma:lnormalgu}\upperspace\\
		& \color{gray}$\phantom{\OhOfExp{D} + {}}w(k, N, \sigma)\phantom{2}$ &
		& \color{gray}$\phantom{\OhOf{\lg k + \lg \ellmax} + {}}t(k, N, \sigma)$ & 
		\textcolor{gray}{--}\lowerspace\\\hline
		\multirow{3}{*}{$\underset{\text{w.h.p.}}{{\text{EREW}}}$} & 
		$\OhOf{k \sqrt{\lg k}\lg\lg k + D}$ & 
		$\factorshow$ & $\OhOf{\lg k + \lg \ellmax} + \OhOf{\lg^{3/2} k \cdot \lg\lg k}$ & 
		\multirow{2}{*}{\cref{lemma:lnormalhan}}\upperspace\\
		&$\phantom{\OhOfExp{D}}{} + w(k,2D, \sigma)$ && ${\phantom{\OhOf{\lg k + \lg \ellmax}}} + t(k,2D, \sigma)$ &\\
		& \color{gray}$\phantom{\OhOfExp{D} + {}}w(k,N, \sigma)$ &
		& \color{gray}${\phantom{\OhOf{\lg k + \lg \ellmax} + {}}}t(k,N,\sigma)$ & 
		\textcolor{gray}{--}\lowerspace\\\hline
	\end{tabularx}
	\label{tbl:results}
\end{table}

\vspace{.5\baselineskip}

The rest of the paper is structured as follows: 
In \cref{sec:prelim} we introduce the basic notation and definitions regarding the PRAM model and string processing. 
In \cref{sec:generalidea} we explain our approximation scheme for the distinguishing prefix, which we use in \cref{sec:derive:determ} to derive deterministic $D$-aware string sorters. 
By using Karp-Rabin fingerprinting, we can also derive randomized string sorters, and achieve better complexity bounds for our approximation scheme (\cref{sec:fingerprints}). 
We summarize our results in \cref{sec:conclude}.

\section{Preliminaries}
\label{sec:prelim}

Throughout this paper, we write $\lg x$ to denote the binary logarithm $\log_2 x$, and $[x, y]$ to denote the discrete interval $\{x, x+1, \dots, y\}$. 
Our research is situated in the PRAM model of computation, where multiple processors work on a shared memory. 
In each processing cycle, each processor may read from a memory cell, write to a memory cell, or perform a simple local operation (logical shifts, basic arithmetic operations etc). 
We consider the following variations of the PRAM model: 
EREW (each memory location can be read and written by at most one processor in each time step), 
CREW (each memory location can be read by multiple processors in each time step, and written by a single processor in each time step), and 
CRCW (each memory location can be read and written by multiple processors in each time step). 
For the CRCW model, we consider two variants: In the Common-CRCW model, multiple processors are allowed to write to the same memory location in the same time step only if all of them write the same value. 
In the Arbitrary-CRCW model, multiple processors are allowed to write different values to the same memory location in the same time step, and an arbitrary processor succeeds. 
However, the designer of an algorithm for this model may not make any assumptions as to which one it is. 
The time required by a PRAM algorithm is the total number of processing cycles. 
The \emph{work} of a PRAM algorithm is defined as the total number of primitive operations that are performed by all processors, or (equivalently) as the running time of the algorithm when using only a single processor.
One of the most fundamental operations in the PRAM model is the \emph{all-prefix-operation}, and its specialization, the \emph{all-prefix-sums-operation}:

\begin{lemma}[All-Prefix-Operation, e.g.\ \cite{blelloch:prefix-sums}]\label{lemma:prefixop}
	Let $a_1, \dots, a_n$ be $n$ integers, and let $\oplus$ be a binary associative operator that can be evaluated in constant time. 
	The sequence 
	$a_1, (a_1 \oplus a_2), (a_1 \oplus a_2 \oplus a_3), \dots, (a_1 \oplus \dots \oplus a_n)$	
	can be computed in the EREW model in $\OhOf{n}$ work, $\OhOf{n}$ space and $\OhOf{\lg n}$ time.
\end{lemma}

\begin{lemma}[All-Prefix-Sums, \cite{Cole1989}]\label{lemma:prefixsum}
	The all-prefix-operation with addition as associative operator can be computed in the Common-CRCW model in $\OhOf{n}$ work, $\OhOf{n}$ space and $\OhOf{\lg n / \lg \lg n}$ time. 
\end{lemma}

Next, we introduce basic string processing notations. A \emph{string} over the \emph{alphabet} $\Sigma$ is a finite sequence of \emph{symbols} from the set $\Sigma = \{1, \dots, \sigma\}$.
We write $\absolute{s}$ to denote the length of a string $s$.
The $x$-th symbol of a string is $s[x]$, while the \emph{substring} from the $x$-th to the $y$-th symbol is denoted as $s[x..y] = s[x]s[x+1]\dots s[y]$. 
The substring $s[1..y]$ is called length-$y$ prefix of $s$.

Given $k$ strings $s_1,\dots,s_k$, the length of the \emph{longest common prefix} of two strings $s_i, s_j$ is defined as $lcp(s_i,s_j) = \max \{\,\ell \mid s_i[1..\ell] = s_j[1..\ell]\,\}$. 
Let $\ell = lcp(s_i, s_j)$. 
We say that $s_i$ is \emph{lexicographically not larger} than $s_j$ and write $s_i \leqlex s_j$, iff either $\ell = \absolute{s_i}$, or $\ell < {\min(\absolute{s_i}, \absolute{s_j})}$ and ${s_i[\ell + 1] < s_j[\ell + 1]}$. 
The strings are in \emph{lexicographical order} iff we have $s_1 \leqlex s_2 \leqlex \dots \leqlex s_k$. 
The \emph{relevant prefix length} of $s_i$ is ${\ell_i = \min(\absolute{s_i}, 1+\max\{\,lcp(s_i, s_j)\mid 1 \leq j \leq k \land j \neq i\,\}})$.
The maximum number of characters that need to be inspected for a single string-to-string comparison is ${\ellmax = \max \{\,\ell_i \mid 1 \leq i \leq k\,\}}$. Finally, the \emph{distinguishing prefix size} of the strings is defined as $D = \sum_{i=1}^{k} \ell_i$, which is the minimum number of characters that need to be inspected in order to lexicographically sort the strings. 

Given $k$ strings of total length $N$ over the alphabet $[1, \sigma]$, let $f(k, N, \sigma)$ be a function indicating the resources (e.g.\ the time or space) needed by an algorithm to perform some task on the strings.
We say that $f$ is \emph{resilient in $N$}
iff multiplying $N$ by a constant factor increases $f$ by at most a constant factor, i.e.,
\begin{equation}
	\forall c_1 : \exists c_2 : \forall k, N, \sigma: f(k, c_1 \cdot N, \sigma) \leq c_2 \cdot f(k, N, \sigma)
\end{equation}

(where all variables are from $\mathbb{N}^+$). This property will be useful when determining the worst-case complexity bounds of our algorithms. Note that the equation holds in the practical case where $f$ is composed of a constant number of polynomial and polylogarithmic terms.

\FloatBarrier
\section{Approximating the Distinguishing Prefix}
\label{sec:generalidea}

In this section, we introduce our framework for $D$-aware parallel string sorting.
The general approach is to approximate the distinguishing prefix, resulting in an array $L$ of size $k$ with $L[i] \in [\ell_i, 2\ell_i)$, i.e.\ we obtain a $2$-approximation of the relevant prefix lengths. 
Afterwards, we can safely prune each string $s_i$ to its prefix $s'_i = s_i[1..L[i]]$. Clearly, the total length of the strings $s'_1, \dots s'_k$ is less than $2D$, and for any two strings we have $s_i \llex s_j \Leftrightarrow s'_i \llex s'_j$. Therefore, we can then use any (not $D$-aware) string sorting algorithm to sort the strings in time and work depending solely on $k$, $D$, and $\sigma$.

Broadly speaking, the approximation scheme performs $\ceil{\lg\ellmax} + 1$ rounds, where in round $r$ we identify and discard the strings $s_i$ with $\ell_i \in (2^{r - 1}, 2^r]$ (starting with round $r = 0$). 
More precisely, amongst all not yet discarded strings, we determine the ones whose length-$2^r$ prefix is unique.
Since any such string has not been discarded in the previous rounds, we have $\ell_i > 2^{r - 1}$, while the uniqueness of the length-$2^r$ prefix guarantees $\ell_i \leq 2^r$.
By assigning $L[i] \gets \min(\absolute{s_i}, 2^r)$, we obtain the desired $2$-approximation of $\ell_i$. The algorithm terminates as soon as all strings have been discarded (and thus all relevant prefix approximations have been found).

\vspace{.5\baselineskip}

\newcommand{\iarr}{I_{\text{arr}}}

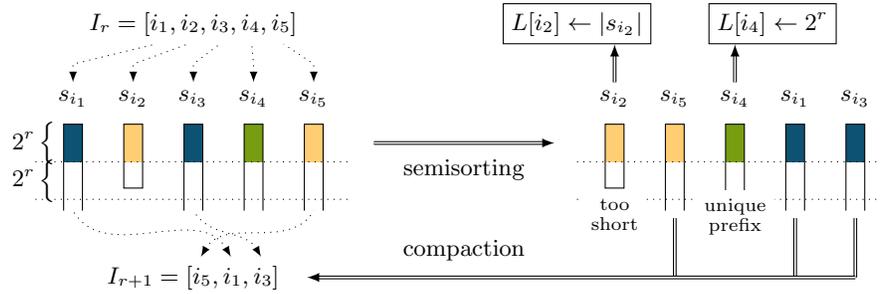
\begin{figure}[t]
	\newlength{\barheight}\setlength{\barheight}{.5cm}
	\newlength{\barwidth}\setlength{\barwidth}{.25cm}
	\newlength{\bardist}\setlength{\bardist}{.8cm}
	\newcommand{\colA}{col1}
	\newcommand{\colB}{col2}
	\newcommand{\colC}{col3}
	\begin{tikzpicture}
		\draw node[inner sep=3pt] (ir1) {$I_r = [i_1, i_2, i_3, i_4, i_5]$};
		\draw node (s1) at ($(ir1.center) - (2\bardist, 4\barwidth)$) {$s_{i_1}$};
		\draw node (s2) at ($(s1.center) + (\bardist, 0)$) {$s_{i_2}$};
		\draw node (s3) at ($(s2.center) + (\bardist, 0)$) {$s_{i_3}$};
		\draw node (s4) at ($(s3.center) + (\bardist, 0)$) {$s_{i_4}$};
		\draw node (s5) at ($(s4.center) + (\bardist, 0)$) {$s_{i_5}$};
		\draw node (s6) at ($(s5.center) + (5\bardist, 0)$) {$s_{i_2}$};
		\draw node (s7) at ($(s6.center) + (\bardist, 0)$) {$s_{i_5}$};
		\draw node (s8) at ($(s7.center) + (\bardist, 0)$) {$s_{i_4}$};
		\draw node (s9) at ($(s8.center) + (\bardist, 0)$) {$s_{i_1}$};
		\draw node (s10) at ($(s9.center) + (\bardist, 0)$) {$s_{i_3}$};
		\foreach \x in {1,...,10} {
			\draw node[below=.0\barwidth of s\x] (mr\x) {};
			\draw node (r\x) at ($(mr\x) - (.5\barwidth, 0)$) {};
			\draw node (br\x) at ($(mr\x) - (0, 2.3\barheight)$) {};
		};
		\fill[\colA] (r1) rectangle ++(\barwidth, -\barheight);
		\fill[\colB] (r2) rectangle ++(\barwidth, -\barheight);
		\fill[\colA] (r3) rectangle ++(\barwidth, -\barheight);
		\fill[\colC] (r4) rectangle ++(\barwidth, -\barheight);
		\fill[\colB] (r5) rectangle ++(\barwidth, -\barheight);
		\fill[\colB] (r6) rectangle ++(\barwidth, -\barheight);
		\fill[\colB] (r7) rectangle ++(\barwidth, -\barheight);
		\fill[\colC] (r8) rectangle ++(\barwidth, -\barheight);
		\fill[\colA] (r9) rectangle ++(\barwidth, -\barheight);
		\fill[\colA] (r10) rectangle ++(\barwidth, -\barheight);
		\foreach \x in {1,3,4,5,7,8,9,10} {
			\draw 
				($(r\x) - (0, 2.3\barheight)$) -- 
				++(0, 2.3\barheight) --
				++(\barwidth, 0) --
				++(0, -2.3\barheight);
		};
		\draw (r2) rectangle ++(\barwidth, -1.7\barheight);
		\draw (r6) rectangle ++(\barwidth, -1.7\barheight);
		\draw[dotted] ($(mr1) - (2\barwidth, \barheight)$) -- ($(mr5) - (-2\barwidth, \barheight)$);
		\draw[dotted] ($(mr1) - (2\barwidth, 2\barheight)$) -- ($(mr5) - (-2\barwidth, 2\barheight)$);
		\draw[dotted] ($(mr6) - (2\barwidth, \barheight)$) -- ($(mr10) - (-2\barwidth, \barheight)$);
		\draw[dotted] ($(mr6) - (2\barwidth, 2\barheight)$) -- ($(mr10) - (-2\barwidth, 2\barheight)$);
		\draw ($(r1) - (-1.5\barwidth, .5\barheight)$) node[left, align=center] {$2^r \begin{cases}\rule{0cm}{.45cm}\end{cases}$};
		\draw ($(r1) - (-1.5\barwidth, 1.5\barheight)$) node[left, align=center] {$2^r \begin{cases}\rule{0cm}{.45cm}\end{cases}$};
		\draw 
			($(mr5) + (\bardist, -.5\barheight)$) 
			edge[double, -latex] node[below=3pt, align=center] (sslabel) {semisorting}
			($(mr6) - (\bardist, .5\barheight)$);
			
		\draw ($(mr3) - (0, 2.3\barheight)$) node (nextI) {};
		\draw node[below=2\barwidth of nextI] (ir2) {$I_{r + 1} = [i_5, i_1, i_3]$};
		
		\draw ($(mr6) - (0pt, 1.75\barheight)$) node[fill=white, below, align=center] 
		{\scriptsize too\\[-.3\baselineskip]\scriptsize short};
		\draw ($(mr8) - (-0pt, 1.75\barheight)$) node[fill=white, below, align=center] 
		{\scriptsize unique\\[-.3\baselineskip]\scriptsize prefix};
		
		\draw ($(mr7) - (0, 2.5\barheight)$) edge[double] (mr7 |- ir2);
		\draw ($(mr9) - (0, 2.5\barheight)$) edge[double] (mr9 |- ir2);
		\draw ($(mr10) - (0, 2.5\barheight)$) edge[double] (mr10 |- ir2);
		\draw (mr10 |- ir2) edge[double, -latex] ($(ir2.east) + (\barwidth, 0)$);
		\draw (sslabel.center |- ir2.center) node[above=3pt] {compaction};
		
		\path (s1) edge[draw=none] node[midway] (mrtop) {} (ir1);		
		\draw 
			($(ir1.south) - (.5cm, 0)$) node (top1) {}
			($(ir1.south) - (.075cm, 0)$) node (top2) {}
			($(ir1.south) + (.35cm, 0)$) node (top3) {}
			($(ir1.south) + (.775cm, 0)$) node (top4) {}
			($(ir1.south) + (1.2cm, 0)$) node (top5) {};
			
		\foreach \x in {1,2,3,4,5} {
			\draw[dotted, latex-] (s\x) to (s\x |- mrtop) to (top\x.center);
		}

		\draw 
			($(ir2.north) + (.5cm, 0)$) node (bot1) {}
			($(ir2.north) + (.9cm, 0)$) node (bot3) {}
			($(ir2.north) + (.1cm, 0)$) node (bot5) {}
			(bot1) edge[draw=none] node[pos=.65] (mrbot1) {} (br1)
			(bot3) edge[draw=none] node[pos=.9] (mrbot3) {} (br3)
			(bot5) edge[draw=none] node[pos=.4] (mrbot5) {} (br5);

		\draw[dotted] (bot1.center) edge[in=300, out=120, latex-] (br1.center);
		\draw[dotted] (bot3.center) edge[in=300, out=120, latex-] (br3.center);
		\draw[dotted] (bot5.center) edge[in=240, out=60, latex-] (br5.center);

		\draw node[draw=black, xshift=-.5cm] (l2) at (s6 |- ir1) {${L[i_2] \gets \absolute{s_{i_2}}}$};
		\draw node[draw=black, xshift=.5cm] (l4)  at (s8 |- ir1) {${L[i_4] \gets 2^r}$};
		
		\draw (s6) edge[double, -latex] (s6 |- l2.south);
		\draw (s8) edge[double, -latex] (s8 |- l4.south);
	\end{tikzpicture}
	\caption{Round $r$ of our approximation scheme. Equal colors identify equal prefixes (best viewed in color).}
	\label{fig:scheme:round}
\end{figure}

Let us look at a single round in technical detail. 
Let $I_r$ be the set of strings (or more precisely their indices) that survived until round $r$, and whose length is at least $2^r$, i.e.\ $I_r = \{i \in [1, k] \mid \ell_i > 2^{r - 1} \land \absolute{s_i} \geq 2^r\}$. 
Initially, before round $r = 0$, we have $I_0 = \{1, \dots, k\}$. 
From now on, let $k_r = \absolute{I_r}$ denote the number of strings that survived until round $r$.
Before starting the round, we assume that $I_r$ is given as a compact array of $k_r$ words. 
Each round consists of two phases, which we explain in the following. 
The description is supported by \cref{fig:scheme:round}.

\begin{description}
	\item[Semisorting Phase.] We semisort $I_r$ using the length-$2^r$ prefixes of the corresponding strings as keys (i.e.\ entry $I_r[j] = i$ is represented by the key $s_i[1..2^r]$). 
	Semisorting is a relaxation of sorting that reorders the entries such that equal keys are contiguous, but different keys do not necessarily have to appear in correct order. 
	In the upcoming sections, we propose different approaches for this phase.

	\item[Compaction Phase.] Let $I_r$ be semisorted as described above, and let $i = I_r[j]$ be any entry. 
	Furthermore, let $i^{-} = I_r[j - 1]$ and $i^{+} = I_r[j + 1]$ be the neighboring entries of $I_r[j]$. 
	Due to the semisorting, the length-$2^r$ prefix of $s_i$ is unique iff $s_{i^{-}}[1..2^r] \neq s_i[1..2^r] \neq s_{i^{+}}[1..2^r]$. 
	We trivially check this condition for all entries simultaneously in $\OhOf{k_r \cdot 2^r}$ work
	and $\OhOf{1}$ time in the Common-CRCW model, or in the same work 
	and $\OhOf{\lg 2^r} = \OhOf{r}$ time in the EREW model (which can be easily achieved using \cref{lemma:prefixop}).
	If the prefix of $s_i$ is unique, we assign $L[i] \gets 2^r$ and $I_r[j] \gets 0$ (where $I_r[j] = 0$ indicates that we no longer need to consider $s_i$ in upcoming rounds).
	Otherwise, we check if $s_i$ is too short to be considered in the next round:
	If $\absolute{s_i} \leq 2^{r + 1}$ holds, we assign $L[i] \gets \absolute{s_i}$ and $I_r[j] \gets 0$.
	Finally, we obtain $I_{r + 1}$ by moving the non-zero entries of $I_r$ to the front of the array. 
	This requires a single all-prefix-sums-operation \cite[Section 3.1]{Vishkin2010}, and thus $\OhOf{k_r}$ work
	and $\OhOf{\lg k_r}$ time in the EREW model, or the same work
	and $\OhOf{\lg k_r / \lg\lg k_r}$ time in the Common-CRCW model (\cref{lemma:prefixop,lemma:prefixsum}).
\end{description}

\noindent\textbf{\textsf{Complexity.}}
Before discussing different approaches for the semisorting phase, we already give general bounds for the work and time complexity of our approximation scheme.
For this purpose we only consider the compaction phase, which takes $\OhOf{k_r \cdot 2^r}$ work in round $r$ (regardless of the PRAM model)
and thus $\OhOf{\sum_{r = 0}^{\infty
} k_r \cdot 2^r}$ work in total. This is asymptotically optimal:
%
\begin{equation}\label{eqn:optimalwork}
	\newcommand{\tmpspace}{\ }
	\sum_{r = 0}^{\mathclap{\ceil{\lg \ellmax}}} k_r \cdot 2^r 
	\tmpspace=\tmpspace
	\sum_{r = 0}^{\mathclap{\ceil{\lg \ellmax}}} \sum_{\ i \in I_r} 2^r 
	\tmpspace\leq\tmpspace
	\sum_{i = 1}^k \sum_{r = 0}^{\ceil{\lg \ell_i}} 2^r 
	\tmpspace<\tmpspace
    \sum_{i = 1}^{k} 2^{\ceil{\lg \ell_i} + 1}
    \tmpspace\leq\tmpspace
    \sum_{i = 1}^{k} 4\ell_i 
	\tmpspace=\tmpspace
	4D
\end{equation}
%

Next, we focus on the execution time in the EREW model. 
The compaction phase of round $r$ takes $\OhOf{r + \lg k_r} \subseteq \OhOf{\lg \ellmax + \lg k}$ time, resulting in $\OhOf{\lg \ellmax \cdot (\lg \ellmax + \lg k)}$ time for all rounds. 
In the stronger Common-CRCW model, we have ${\OhOf{\lg k_r / \lg\lg k_r} \subseteq \OhOf{\lg k / \lg \lg k}}$ time for round $r$, and thus $\OhOf{\lg d \cdot \lg k / \lg \lg k}$ time in total.


\FloatBarrier

\section{Deriving Deterministic $D$-aware String Sorters}

\label{sec:derive:determ}

The perhaps easiest solution for the semisorting phase is to use an existing string sorter as a subroutine, e.g.\ one of the algorithms that we discussed in \cref{sec:relwork}. 
Then, after finishing the last round of our approximation scheme, we reduce the strings to their length-$L[i]$ prefixes and sort them with the same algorithm that we already used during the semisorting phase.
This naturally results in a new $D$-aware string sorter, as visualized in \cref{fig:derive}.

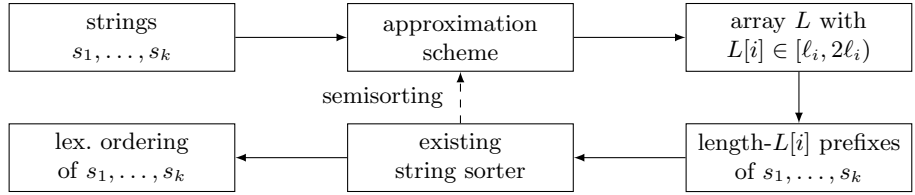
\begin{figure}[t]
	\small
	\begin{tikzpicture}
		\tikzset{mn/.style={draw, align=center, minimum width=3cm, minimum height=.9cm}}
		\draw 
			node[mn] (a) {strings\\$s_1, \dots, s_k$}
			(a.center) ++(4.5cm, 0) node[mn] (b) {approximation\\scheme}
			(b.center) ++(4.5cm, 0) node[mn] (c) {array $L$ with\\$L[i] \in [\ell_i, 2\ell_i)$}
			(c.center) ++(0, -1.6cm) node[mn] (d) {length-$L[i]$ prefixes\\of $s_1, \dots, s_k$}
			node[mn] (e) at (b |- d) {existing\\string sorter}
			node[mn] (f) at (a |- d) {lex.\ ordering\\of $s_1, \dots, s_k$}
			(a) edge[-latex] (b) 
			(b) edge[-latex] node[pos=.6] (x) {} (c) 
			(c) edge[-latex] (d) 
			(d) edge[-latex] (e) 
			(e) edge[-latex] (f);
		\draw (e) edge[-latex, dashed] node[left] (y) {semisorting\ \ } (b);
	\end{tikzpicture}
	\caption{Deriving $D$-aware string sorters from existing $D$-unaware solutions.}
	\label{fig:derive}
\end{figure}

We obtain a general result for an important class of sorters: The ones that do not rely on comparison sorting and typically require $N \cdot w(k, N, \sigma)$ work and $t(k, N, \sigma)$ time for some functions $w$ and $t$ that are resilient in $N$ and non-decreasing in $k$ and $N$ (e.g.\ Hagerup's algorithm \cite{DBLP:conf/stoc/Hagerup94}).
Using such an algorithm, the semisorting phase of round $r$ takes $(k_r \cdot 2^r) \cdot w(k_r, k_r \cdot 2^r, \sigma)$ work.
Summing up all rounds, the total work for semisorting is $\OhOf{D \cdot w(k, D, \sigma)}$:
\begin{equation}
	\newcommand{\tmpspace}{\ \ }
	\sum_{r = 0}^{\mathclap{\ceil{\lg \ellmax}}} (k_r \cdot 2^r) \cdot w(k_r, k_r \cdot 2^r, \sigma)
	\tmpspace\leq\tmpspace
	\sum_{r = 0}^{\mathclap{\ceil{\lg \ellmax}}} (k_r \cdot 2^r) \cdot w(k, 2D, \sigma)
	\tmpspace<\tmpspace
	4D \cdot w(k, 2D, \sigma)
\end{equation}

The first inequality holds because $w$ is non-decreasing in $k$ and $N$, while the second one holds due to \cref{eqn:optimalwork}. 
We have $w(k, 2D, \sigma) = \OhOf{w(k, D, \sigma)}$ because $w$ is resilient in $N$. 
For the same reason, the time for the semisorting phase of round $r$ is $t(k_r, k_r \cdot 2^r, \sigma) \leq t(k, 2D, \sigma) = \OhOf{t(k, D, \sigma)}$. 
Combined with the bounds from \cref{sec:generalidea} we have:
\begin{theorem}\label{lemma:sortersubroutine}
	Let $N \cdot w(k, N, \sigma)$ and $t(k, N, \sigma)$ be the work and time needed by some algorithm to sort $k$ strings of total length $N$ over the alphabet $[1, \sigma]$ (for arbitrarily large $\sigma$), with $w$ and $t$ resilient in $N$ and non-decreasing in $k$ and $N$.
    Let $D$ be the distinguishing prefix size. 
    Then we can sort the strings in $\OhOf{D \cdot w(k, D, \sigma)}$ work and
	$\OhOf{\lg d \cdot (\lg d + \lg k + t(k, D, \sigma))}$ time. The PRAM model matches the one of the string sorter.  If the model is at least as strong as the Common-CRCW model, the time decreases to $\OhOf{{\lg d \cdot (\lg k / \lg \lg k + t(k, D, \sigma))}}$.
\end{theorem}

Note that the theorem requires a string sorter that allows \emph{arbitrary} alphabets. This is due to the fact that (even after the first round) the number $k_r$ of remaining strings can become arbitrarily small. Consequently, the alphabet size might become arbitrarily large compared to the total length $k_r \cdot 2^r$ of the strings that we have to semisort in round $r$.


\subsubsection{Dealing With Large Alphabets.}

In theory, \cref{lemma:sortersubroutine} directly implies new $D$-aware string sorters. 
However, while the theorem applies to sorters for arbitrary alphabets, many of the existing string sorting algorithms are restricted to polynomial alphabets (i.e.\ $\sigma = N^{\OhOf{1}}$).
In the remainder of this section, we show that even such alphabet restricted sorters work with \cref{lemma:sortersubroutine}, if we equip them with an additional preprocessing routine.
We demonstrate the technique using Hagerup's algorithm \cite{DBLP:conf/stoc/Hagerup94} as an example.
It will be easy to see that it would just as well work with any other string sorter. 
Recall Hagerup's original result:

\begin{lemma}[Hagerup \cite{DBLP:conf/stoc/Hagerup94}, Theorems 4.4 and 4.5]
	A set of strings of total length $N$ over the alphabet $[1, N^{\OhOf{1}}]$ can be sorted in $\OhOf{\lg N / \lg \lg N}$ time and $\OhOf{N\lg \lg N}$ work in the CRCW model, or in $\OhOf{N \sqrt{\lg N}}$ work and $\OhOf{\lg^{3/2} N\sqrt{\lg\lg N}}$ time in the CREW model, or in $\OhOf{N \sqrt{\lg N\lg\lg N}}$ work and $\OhOf{\lg^{3/2} N\sqrt{\lg\lg N}}$ time in the EREW model.
	
	\vspace{.5\baselineskip} 
	
	\noindent\textit{Remark:} \textnormal{Hagerup does not explicitly state which variant of the CRCW model is used. 
	However, the algorithm relies on a padded integer sorting subroutine that requires the Arbitrary-CRCW model \cite{Hagerup1993}. It appears that all other operations performed by the algorithm require at most the Arbitrary-CRCW model as well.}
	\label{lemma:hageruppoly}
\end{lemma}

In order to generalize \cref{lemma:hageruppoly} to arbitrary alphabets $[1, \sigma]$ with $\sigma \notin N^{\OhOf{1}}$, we perform a preprocessing that reduces the alphabet to $[1, N]$ in an order preserving manner.
The general idea is to use an integer sorter to sort the symbols that actually occur in any of the strings.
Then, we can simply replace each symbol with its rank amongst the sorted symbols.
A similar reduction technique has previously been used by Hagerup \cite[p. 389]{DBLP:conf/stoc/Hagerup94} (but for a different purpose). 
For now, we only consider the Arbitrary-CRCW model.

First, we create $N$ tuples of the form $\tuple{i, j, c}$, where $c$ is the $j$-th symbol of $s_i$.
{%
\let\oldcdot\cdot%
\renewcommand{\cdot}{{\,\oldcdot\,}}%
Initially, the tuples are ordered by their first and second component, i.e.\ $\tuple{1, 1, \cdot} \dots \tuple{1, \absolute{s_1}, \cdot} \tuple{2, 1, \cdot} \dots \tuple{2, \absolute{s_2}, \cdot} \dots \tuple{k, 1, \cdot} \dots \tuple{k, \absolute{s_k}, \cdot}$.
}
In order to store this sequence in a consecutive memory area, we have to determine the position of each tuple within the sequence.
Using the all-prefix-sums-operation, we can trivially realize this step in $\OhOf{N}$ work and $\OhOf{\lg N / \lg \lg N}$ time due to \cref{lemma:prefixsum}.
Then, we use the integer sorting algorithm by Bhatt et al.\cite{Bhatt1991} to sort the tuples by their third component, which takes $\OhOf{N \lg \lg \sigma}$ work and $\OhOf{\lg N / \lg \lg N + \lg\lg\sigma}$ time.
Let $\tuple{i_1, j_1, c_1} \dots \tuple{i_N, j_N, c_N}$ be the sorted sequence of tuples.
In an array $A \in \{0, 1\}^N$, we mark the (in terms of the sequence) leftmost occurrence of each character, i.e.\ $\forall h \in [2, N] : c_{h - 1} \neq c_h \Leftrightarrow A[h] = 1$.
Next, we replace $A$ with its prefix-sums, once again taking $\OhOf{N}$ work and $\OhOf{\lg N / \lg \lg N}$ time due to \cref{lemma:prefixsum}.
Now each entry $A[h]$ contains exactly the rank of the symbol $c_h$ amongst all symbols.
Finally, for each $h \in [1, N]$, we replace the $j_h$-th symbol of the $i_h$-th string with $A[h] + 1$. 
Since this reduces the alphabet to (a subset of) $[1, N]$ in an order preserving manner, we can sort the strings using \cref{lemma:hageruppoly}. 

In the weaker CREW and EREW models we use the same technique, but replace the algorithm by Bhatt et al.\ with Han and Shen's integer sorter in the EREW model \cite[Theorem 4.1]{Han2002}, which sorts the $N$ tuples in $\OhOf{N\sqrt{\lg N}}$ work and $\OhOf{\lg^{3/2}N}$ time. We have shown:

\begin{corollary}
	A set of strings of total length $N$ over the alphabet $[1, \sigma]$ can be sorted in $\OhOf{\lg N / \lg \lg N + \lg\lg\sigma}$ time and $\OhOf{N\lg \lg N + N \lg \lg \sigma}$ work in the Arbitrary-CRCW model, 
	or in $\OhOf{N \sqrt{\lg N}}$ work and $\OhOf{\lg^{3/2} N\sqrt{\lg\lg N}}$ time in the CREW model, or in $\OhOf{N \sqrt{\lg N\lg\lg N}}$ work and $\OhOf{\lg^{3/2} N\sqrt{\lg\lg N}}$ time in the EREW model.	
	\label{lemma:hagerupgeneral}
\end{corollary} 


\begin{theorem}\label{lemma:hagerupfinal}
	A set of $k$ strings over the alphabet $[1, \sigma]$ with distinguishing prefix size $D$ and longest relevant prefix of length $d$ can be sorted in the work and time stated in \cref{tbl:results}(a).
\end{theorem}

The theorem follows from \cref{lemma:hagerupgeneral} and \cref{lemma:sortersubroutine}.
Note that the work and time in the Arbitrary-CRCW model are $\OhOf{D \lg \lg D}$ and $\OhOf{\lg d \cdot \lg D / \lg \lg D}$, respectively, if the alphabet is quasipolynomial in the distinguishing prefix size, i.e.\ $\sigma = D^{(\lg^{\OhOf{1}} D)}$.

\subsection{Deriving Comparison-Based Sorters}
\label{sec:derive:comp}

As mentioned earlier, any comparison-based string sorter requires $\Omega(k \lg k + D)$ work. 
In this section, we take the $\OhOf{k \lg k + N}$ work algorithm by J\'aJ\'a et al.\ \cite{DBLP:journals/tcs/JaJaRV96}, and derive an $\OhOf{k \lg k + D}$ work modification, thus matching the lower bound.
Assuming that we use the $\OhOf{k \lg k + N}$ work algorithm to realize the semisorting phase of our approximation scheme, the work for semisorting in round $r$ becomes $\OhOf{k_r \lg k_r + k_r \cdot 2^r}$. After the $\ceil{\lg \lg k}$-th round, the $k_r \cdot 2^r$ term 
dominates the $k_r \lg k_r$ term. Therefore, the total work for semisorting is:
\begin{equation}
	\OhOfLarge{\sum_{r = 0}^{\ceil{\lg \ellmax}} k_r \lg k_r + k_r \cdot 2^r} 
	\enskip=\enskip
	\OhOfLarge{\sum_{r = 0}^{\ceil{\lg \lg k}-1} k_r \lg k_r + \sum_{r = 0}^{\ceil{\lg \ellmax}} k_r \cdot 2^r}
\end{equation}

Following \cref{eqn:optimalwork}, the second sum on the right-hand side of the equation is bounded by $\OhOf{D}$. Unfortunately, there appears to be no such upper bound for the first sum.
Therefore, we relax our approximation scheme by simply skipping the initial $\ceil{\lg \lg k}$ rounds. 
This way, the first round that we actually perform is round $r = \ceil{\lg \lg k}$, during which we consider prefixes of length $2^{\ceil{\lg \lg k}} < 2\lg k$. 
Note that consequently we may overestimate the length of relevant prefixes by $2\lg k$ additional symbols, i.e.\ we obtain $\Lrelax[i] \in [\ell_i, 2 \cdot \max(\lg k, \ell_i))$. Thus, when truncating each string to its prefix $s_i[1..\Lrelax[i]]$, the total length of the strings is
\begin{equation}
	D' := \sum_{i = 1}^k \Lrelax[i] < 2\sum_{i = 1}^k (\lg k + \ell_i) = 2 k \lg k + 2D.
\end{equation}

Therefore, after computing $\Lrelax$, we can use the algorithm by J\'aJ\'a et al.\ once more to sort the truncated strings in optimal $\OhOf{k \lg k + D'} \subseteq \OhOf{k \lg k + D}$ work. 
The semisorting in round $r$ takes $\OhOf{\lg^2 k_r / \lg \lg k_r} \subseteq \OhOf{\lg^2 k / \lg \lg k}$ time, and there are $\ceil{\lg d} - \ceil{\lg \lg k} = \OhOf{\lg d}$ rounds. 
Together with the bounds from \cref{sec:generalidea} we have:

\begin{theorem}\label{lemma:jajaplain}
	A set of $k$ strings with distinguishing prefix size $D$ and longest relevant prefix of length $d$ can be sorted in the Common-CRCW model in $\OhOf{k \lg k + D}$ work and $\OhOf{\lg d \cdot \lg^2 k / \lg \lg k}$ time.
\end{theorem}

Note that we cannot trivially use our approximation scheme to derive a $D$-aware modification of the randomized string sorter by J\'aJ\'a et al. \cite{DBLP:journals/tcs/JaJaRV96}, which sorts $k$ strings of total length $N$ in $\OhOf{k \lg k + N}$ work and $\OhOf{\lg k}$ time with high probability, i.e.\ with probability $1 - (1/k)^{c}$ for any constant $c > 0$. 
If we were using this algorithm for the semisorting phase, then the probability of successfully sorting the remaining strings in round $r$ would be $1 - (1/{k_r})^c$. 
However, even after the first round, $k_r$ can become arbitrarily small, resulting in a low probability of success. 
The randomized string semisorters from the next section will allow us to circumvent this problem.

\section{Randomized String Semisorting}
\label{sec:fingerprints}

In this section, we equip our approximation scheme with randomized string semisorters that are based on Karp-Rabin fingerprints~\cite{DBLP:journals/ibmrd/KarpR87}. 
The goal of these fingerprints is to hash substrings to small integers, which allows fast equality testing. 
Consider the semisorting phase of round~$r$, during which we have to semisort $k_r$ string prefixes of length $2^r$ each. 
Instead of directly semisorting the prefixes, we first compute a fingerprint as a representative for each prefix, and then semisort the fingerprints. 
This way, we can use less complex integer sorting algorithms as a subroutine.
Before going into detail, we show how to efficiently compute fingerprints in the EREW model.

In order to define Karp-Rabin fingerprints, we use a prime number $q = \Theta(N^{c})$ for some constant $c > 1$, and a value $b \in [q, 2q)$ chosen uniformly at random. 
The Karp-Rabin fingerprint $\phi_i(x,y)$ of the substring $s_i[x..y]$ is defined as follows:
\begin{equation}
\phi_i(x,y) = \sum_{z=x}^y s_i[z] \cdot b^{y-z} \mod q
\end{equation}

Observe that equal substrings have equal fingerprints, i.e.\ for every integer $n \ge 0$ it holds $s_i[x..x+n] = s_j[y..y+n] \implies \phi_i(x,x+n) = \phi_j(y,y+n)$. 
On the other hand, if two substrings are not equal, their fingerprints will be different with high probability. 
In particular, if $s_i[x..x+n] \ne s_j[y..y+n]$ then $\textsf{Prob}[\phi_i(x,x+n) = \phi_j(y,y+n)] \le \frac{n+1}{q} = \OhOf{N^{1-c}}$. Thus, by choosing a large enough constant $c>1$, we can control the probability of false positives when comparing fingerprints instead of substrings. 
Using the all-prefix-operation, Karp-Rabin fingerprints can be computed efficiently in parallel:

\begin{lemma}\label{lemma:fingerprint}
\let\oldell\ell
\renewcommand{\ell}{\oldell}
For every $\ell$-character substring $s_i[x..x+\ell - 1]$, the Karp-Rabin fingerprint $\phi_i(x,x+\ell-1)$ can be computed in $\OhOf{\ell}$ work, $O(\ell)$ space, and $\OhOf{\lg \ell}$ time in the EREW model.
\end{lemma}
\begin{proof}
First, we compute the sequence of exponents $b^0, b^1, \dots, b^{\ell-1} \pmod q$ using the all-prefix-operation with multiplication over $\Zq$ as the associative \mbox{operator}. 
Then, we simultaneously compute all values $f_0, \dots, f_{\ell - 1}$ with $f_j = s_i[x + j] \cdot b^{\ell - j - 1} \pmod q$ in constant time.
Finally, the Karp-Rabin fingerprint $\phi_i(x, x+\ell-1)$ is the sum of all the $f_j$ over $\Zq$, which can be computed via another all-prefix-operation.
The stated complexity bounds follow from \cref{lemma:prefixop}.\qed
\end{proof}

During round $r$ of our approximation scheme, we can simultaneously compute the fingerprints of all length-$2^r$ prefixes, which takes $\OhOf{k_r \cdot 2^r}$ work and ${\OhOf{r} \subseteq \OhOf{\lg d}}$ time. 
It remains to be shown how to semisort the fingerprints.
For now, similarly to \cref{sec:derive:comp}, we skip the first $\ceil{\lg \lg k}$ rounds. 
In the remaining rounds, we use Cole's parallel merge sort \cite{cole:mergesort}, which sorts the $k_r$ fingerprints in round $r$ in $\OhOf{k_r \lg k_r} \subseteq \OhOf{k_r \cdot 2^r}$ work and $\OhOf{\lg k_r}$ time.
This results in the following complexity bounds:

\begin{lemma}\label{lemma:lrelax}
	For any constant $c > 0$, the array $\Lrelax$ with $\Lrelax[i] \in [\ell_i, 2 \cdot \max(\lg k, \ell_i))$ can be computed in the EREW model in $\OhOf{D}$ work and $\OhOf{\lg d \cdot (\lg d + \lg k)}$ time w.h.p.\ $1 - (1/N)^c$.
\end{lemma}

Now we can already derive a $D$-aware modification of the randomized string sorter by J\'aJ\'a et al \cite{DBLP:journals/tcs/JaJaRV96}. 
Just as in \cref{sec:derive:comp}, we simply compute $\Lrelax$ (using \cref{lemma:lrelax}), and then run the original string sorter. 
It follows:

\begin{theorem}\label{lemma:jajarand}
	For any constant $c > 0$, a set of $k$ strings with distinguishing prefix size $D$ and longest relevant prefix of length $d$ can be sorted in the Common-CRCW model in $\OhOf{k \lg k + D}$ work and $\OhOf{\lg d \cdot (\lg d + \lg k)}$ time w.h.p.\ $1 - (1/k)^c$.
\end{theorem}

\subsection{Handling the Initial \boldmath$\ceil{\lg\lg k}$ \unboldmath Rounds}

Finally, we show how to (semi-)sort the fingerprints in the first $\ceil{\lg\lg k}$ rounds. 
Ideally, we would like to use the randomized semisorter by Gu et al.\cite{DBLP:conf/spaa/GuSSB15}, which sorts $k_r$ fingerprints in the Arbitrary-CRCW model in expected optimal $\OhOf{k_r}$ work and $\OhOf{\lg k_r}$ time with high probability $1 - (1/k_r)^c$. 
However, as in the previous section, $k_r$ and thus the probability of success can become arbitrarily small. 
Therefore, we only use the semisorter by Gu et al.\ in rounds when $k_r > k / \lg^2 k$ (resulting in $\OhOf{k_r}$ work), and Cole's mergesort, otherwise (resulting in $\OhOf{k / \lg k}$ work). 
This way, in every round the expected work for semisorting fingerprints is $\OhOf{k_r + k / \lg k}$, the time is $\OhOf{\lg k}$, and the probability of success is at least $1 - (\lg^2 k/k)^c > 1 - (1/k)^{(c/2)}$. 
Summing up the expected work for semisorting during the first $\ceil{\lg \lg k}$ rounds, we have:

\begin{equation*}
	\sum_{r = 1}^{\ceil{\lg \lg k}} k_r + \sum_{r = 1}^{\ceil{\lg \lg k}} k / \lg k = \sum_{r = 1}^{\ceil{\lg \lg k}} k_r + o(k) = \OhOf{D}.
\end{equation*}

Together with the bounds for computing fingerprints (see \cref{sec:fingerprints}) and for the compaction phase (see \cref{sec:generalidea}), we get:

\begin{lemma}\label{lemma:lnormalgu}
	For any constant $c > 0$, the array $L$ with $L[i] \in [\ell_i, 2\ell_i)$ can be computed in the Arbitrary-CRCW model in expected optimal $\OhOf{D}$ work and $\OhOf{\lg d \cdot (\lg d + \lg k)}$ time w.h.p.\ $1 - (1/k)^c$.
\end{lemma}

In the weaker EREW model, we can replace the semisorter by Gu et al.\ with the deterministic integer sorter by Han and Shen \cite{Han2002} that we already used in the proof of \cref{lemma:hagerupgeneral}. This results in the following bounds:

\begin{lemma}\label{lemma:lnormalhan}
	For any constant $c > 0$, the array $L$ with $L[i] \in [\ell_i, 2\ell_i)$ can be computed in the EREW model in $\OhOf{k \sqrt{\lg k}\lg\lg k + D}$ work and $\OhOfB{\lg d \cdot (\lg d + \lg k) + \lg^{3/2} k \cdot \lg \lg k}$ time w.h.p.\ $1 - (1/N)^c$.
\end{lemma}

Note that the probability of success is $1 - (1/N)^c$ (rather than $1 - (1/k)^c$ as in \cref{lemma:lnormalgu}) because we no longer use a probabilistic semisorter, and errors can only occur due to fingerprint collisions.

\cref{lemma:lnormalgu,lemma:lnormalhan} directly imply the results stated in \cref{tbl:results}(c).

\section{Conclusion and Open Questions}
\label{sec:conclude}

We presented a theoretical framework that approximates the distinguishing prefix, resulting in the first $D$-aware string sorters in the PRAM model. 
It remains an open question, if the $\lg d$ time factor can be avoided without increasing the work. 
Generally, it is unknown if a constant approximation of the distinguishing prefix can be computed \emph{deterministically} in optimal $\OhOf{D}$ work and reasonable time.

\bibliographystyle{splncs04}
\bibliography{lcp-aware-pss}

\end{document}